\documentclass[a4paper,leqno,11pt]{article}

\usepackage{needspace}
\usepackage[latin1]{inputenc}
\usepackage{amsmath,amssymb,amsthm,
verbatim,
fancyhdr,amsfonts
}

\usepackage{tikz}
\usetikzlibrary{decorations.pathreplacing}

\usepackage{amsmath}
\usepackage{amsthm}
\usepackage{amsfonts,amsmath,amssymb}
\usepackage{enumerate,verbatim,color}
\usepackage{epsfig}
\usepackage{textcomp}
\usepackage{graphicx}
\usepackage{amsfonts}
\usepackage{mathrsfs}
\usepackage{upgreek}

\newcommand{\virgolette}[1]{``#1''}
\newtheorem{teorema}{Theorem}[section]

\newtheorem{lemma}[teorema]{Lemma}
\newtheorem{prop}[teorema]{Proposition}
\newtheorem{osss}[teorema]{Remark}

\theoremstyle{definition}

\theoremstyle{remark}
\newtheorem*{assumptionI}{\bf Assumptions I}

\newcommand{\iii}{{\, \vert\kern-0.25ex\vert\kern-0.25ex\vert\, }}

\newcommand{\ffi}{\varphi}

\newcommand{\Z}{\mathbb{Z}}
\newcommand{\N}{\mathbb{N}}
\newcommand{\R}{\mathbb{R}}

\newcommand{\C}{\mathbb{C}}

\newcommand{\Hi}{\mathscr{H}}

\newcommand{\Gi}{\mathscr{G}}
\newcommand{\dd}{\partial}

\newcommand{\ra}{\rangle}
\newcommand{\la}{\langle}

\newcommand{\G}{\Gamma}

\date{}
\title{Controllability of bilinear quantum systems in explicit times via explicit control fields}

\author{Alessandro Duca
	\\ \ \\
	{\small  Institut Fourier, Université Grenoble Alpes}\\
	{\small 100 Rue des Mathématiques, 38610 Gières, France} \\
	{\small \texttt{alessandro.duca@univ-grenoble-alpes.fr}}\\
	{\small ORCID: 0000-0001-7060-1723}
}

\begin{document}

\maketitle
\begin{abstract}
We consider the bilinear Schr\"odinger equation on a bounded one-dimensional domain and we provide explicit times such that the global exact controllability is verified. In addition, we show how to construct controls for the global approximate controllability. 
\end{abstract}

\medskip


\medskip


\section{Introduction}
In non relativistic quantum mechanics any pure state of a closed system is mathematically represented by a wave function $\psi$ in the unit sphere of a Hilbert space $\Hi$. 
We consider the evolution of a particle confined in a one dimensional bounded region and subjected to an external electromagnetic field that plays the role of a control. A standard choice for such a setting is $\Hi=L^2((0,1),\C)$, while the field is represented by an operator $B$ and by a real function $u$, which accounts its intensity. In this framework, the evolution of $\psi$ is modeled by the bilinear Schr\"odinger equation
\begin{equation}\tag{BSE}\label{mainx1}\begin{split}
\begin{cases}
i\dd_t\psi(t)=A\psi(t)+u(t)B\psi(t),\ \ \ \ \ \ \ \ \ \ \ \ t\in(0,T),\  T>0.\\
\psi(0,x)=\psi^0(x).\\
\end{cases}
\end{split}
\end{equation}
The operator $A=-\Delta$ is the Laplacian with Dirichlet homogeneous boundary conditions ($D(A)=H^2\cap H^1_0$), $B$ is a bounded symmetric operator, $u\in L^2((0,T),\R)$ is a control function and $\psi^0$ the initial state of the system. We call $\G^u_t$ the unitary propagator of the $(\ref{mainx1})$ when it is defined.

\medskip
A natural question of practical implications is whether, given any couple of states, there exists $u$ steering the quantum system from the first one to the second. The bilinear Schr\"odinger equation is said to be exactly controllable when the dynamics precisely reaches the target. We denote it approximately controllable when it is possible to approach the target as close as desired.
The $(\ref{mainx1})$ is said {simultaneously controllable} when more initial states are controllable (exactly or approximately) at the same time with the same $u$.

\medskip

The controllability of the bilinear Schr\"odinger equation has already been studied in the literature and we start by mentioning $\cite{ball}$ by Ball, Mardsen and Slemrod. This seminal work on bilinear systems shows the well-posedess of the equation in $\Hi$ when $u\in L^1_{loc}(\R)$ and an import non-controllability result. In particular, it ensures that the attainable set
$$Z(\psi_0):=\{\psi\in \Hi|\ \exists T>0,\ \exists r>1,  \exists u \in L^r_{loc}((0,T),\R):\ \psi=\G^T_u\psi_0\}$$ from any initial state $\psi_0$ in the unit sphere $S$ of $\Hi$ is contained in a countable union of compact sets. Therefore, $Z(\psi_0)$ has dense complement in $S$ and the $(\ref{mainx1})$ is not exactly controllable in $\Hi$. For this reason, weaker notions of controllability have been used in order to deal with this equation.

\medskip
For instance in $\cite{laurent}$, Beauchard and Laurent prove the well-posedness and the {local exact controllability} of the $(\ref{mainx1})$ in $H^s_{(0)}:=D(A^{\frac{s}{2}})$ for $s= 3$, when $B$ is a multiplication operator for suitable $\mu\in H^3((0,1),\R)$. 

\noindent
In $\cite{morgane1}$, Morancey proves the {simultaneous local exact controllability} of two or three $(\ref{mainx1})$ in $H^3_{(0)}$ for suitable operators $B=\mu\in H^3((0,1),\R)$.

\noindent
In $\cite{morganerse2}$, Morancey and Nersesyan extend the previous result. They achieve the {simultaneous global exact controllability} of finitely many $(\ref{mainx1})$ in $H^4_{(0)}$ for a wide class of multiplication operators $B=\mu$ with $\mu\in H^4((0,1),\R)$.

\noindent
In \cite{mio1}, Duca (or the author) proves the {simultaneous global exact controllability in projection} of infinite $(\ref{mainx1})$ in $H^3_{(0)}$ for bounded symmetric operators $B$.

\medskip

Global approximate controllability results for the bilinear Schr\"odinger equation are provided with different techniques. Adiabatic arguments are considered by Boscain, Chittaro, Gauthier, Mason, Rossi and Sigalotti in $\cite{ugo2}$ and $\cite{ugo3}$.
Controllability results are achieved with Lyapunov techniques by Mirrahimi in $\cite{milo}$ and by Nersesyan in $\cite{nerse2}$. 
Lie-Galerking arguments are used by Boscain, Boussa\"id, Caponigro, Chambrion, Mason and Sigalotti in $\cite{chambrion}$, $\cite{chambrion1}$, $\cite{nabile}$ and $\cite{ugo}$.

\medskip
Most of the existing results focus their efforts on proving the exact controllability of the bilinear Schr\"odinger equation without precising the relative controls and times. In order to exhibit those elements, it is necessary to develop new techniques leading to the local exact controllability. Indeed, the common approach does not provide explicit neighborhoods where the result is valid. As a consequence, when the outcome is extended to the global controllability, any track of the dynamics time and of the corresponding control is lost.
To this purpose, we prove the local exact controllability for specific neighborhoods and times. The result leads to the global exact controllability with explicit times and partially explicit control functions.

\medskip

In more technical terms, the main novelties of the work are the following. First, for any suitable couple of eigenfunctions $\phi_j$ and $\phi_k$ of $A$, we construct controls and times such that the relative dynamics of the $(\ref{mainx1})$ drives $\phi_j$ close to $\phi_k$ as much desired with respect to the $H^3_{(0)}-$norm.
Second, we estimate a neighborhood of $\phi_k$ in $H^3_{(0)}$ where the local exact controllability is satisfied in a given time.
Third, by gathering the two previous results, we define a dynamics steering any eigenstate of $A$ to any other in an explicit time.
In conclusion, we apply the proved results to an example.

\medskip
The work represents a contribution to the application of the control theory to the physical systems modeled by the bilinear Schr\"odinger equation. Nevertheless, many improvements are still required and the provided estimates are far from being optimal. For example, in Section $\ref{esempio}$, we consider an electron trapped in an one-dimensional guide of length $\sim 10^{-3}$ meters and subjected to an external electromagnetic field. We show a suitable control field driving the state of the electron from the first excited state to the ground state in a time $T\sim 10^{116}$ seconds. The achieved time is way too large for any practical implementation, however future optimization may lead to more reasonable estimates as we explain afterwards.

\subsection{Framework and main results}\label{frame}
Let us consider the $(\ref{mainx1})$ in the Hilbert space $\Hi=L^2((0,1),\C)$ with 
$$D(A)=H^2((0,1),\C)\cap H^1_0((0,1),\C),\ \ \ \ \ A\psi=-\Delta\psi,\ \ \ \ \ \forall \psi\in D(A).$$
We denote $\la\cdot,\cdot\ra$ the scalar product in $\Hi$ 
and $\|\cdot\|$ the corresponding norm. Let $\{\phi_j\}_{j\in\N^*}$ be an orthonormal basis composed by eigenfunctions of $A$ associated to the eigenvalues $\{\lambda_j\}_{j\in\N^*}$ ($\lambda_k=\pi^2 k^2$) and \begin{equation}\label{eigen}
\phi_j(t)=e^{-iAt}\phi_j=e^{-i\lambda_j t}\phi_j.\end{equation}

\noindent
For $s>0$, we define $h^s(\C)=\Big\{\{x_j\}_{j\in\N^*}\subset{\C}\big|\ \sum_{j=1}^{\infty}|j^s x_j|^2<\infty\Big\}$ equipped with the norm $\Big\|\{x_j\}_{j\in\N^*}\Big\|_{(s)}=\Big(\sum_{j=1}^\infty|j^s x_j|^2\Big)^\frac{1}{2}$ for every $\{x_j\}_{j\in\N^*}\in h^s$ and
$$H^s_{(0)}=H^s_{(0)}((0,1),\C):=D(A^{\frac{s}{2}}),\ \ \ \ \ \ \ \ \|\cdot\|_{(s)}=\Big(\sum_{k=1}^\infty|(\pi k)^s\la\phi_k,\cdot\ra|^2\Big)^\frac{1}{2}.$$
\medskip
Let $H_0^m:=\{\psi \in H^m:\ \dd_x^j\psi(0)=\dd_x^j\psi(1)=0,\ \forall j\in\N^*, \ j\leq m\}$ with $H^m:=H^m((0,1),\C)$ and $m\in\N^*$. We underline that $H^2_{(0)}=H^2\cap H^1_0$. 

\medskip

For two Banach spaces $X$ and $Y$, we denote $L(X,Y)$ the space of the linear bounded operators mapping $X$ in $Y$ and equipped with the norm $\iii \cdot\iii_{L(X,Y)}$. In addition, for $s>0,$ we call $\iii \cdot \iii:=\iii \cdot\iii_{L(\Hi,\Hi)}$ and $$\iii \cdot \iii_{(s)}:=\iii \cdot\iii_{L(H^s_{(0)},H^s_{(0)})},\ \  \ \ \iii \cdot \iii_3:=\iii \cdot\iii_{L(H^3_{(0)},H^3\cap H^1_{0})}.$$
We consider $H^3\cap H^1_{0}$ equipped with the norm $\|\cdot\|_{H^3\cap H^1_{0}}=\sqrt{\sum_{j=1}^3\|\dd_x^j\cdot\|^2}.$

\needspace{3\baselineskip}

\begin{assumptionI}
	The bounded symmetric operator $B$ satisfies the following conditions.
	\begin{enumerate}
		\item For every $k\in\N^*$, there exists $C_k>0$ such that $|\la\phi_j,B\phi_k\ra|\geq\frac{C_k}{j^3}$ for every $j\in\N^*$.
		\item $Ran(B|_{D(A)})\subseteq D(A)$ and $Ran(B|_{H_{(0)}^3})\subseteq H^3\cap H^1_{0}.$
	\end{enumerate}
\end{assumptionI}

For $B_{j,k}:=\la\phi_j,B\phi_k\ra$ with $j,k\in\N^*$, we have $\{B_{j,k}\}_{j\in\N^*},\{B_{j,k}\}_{k\in\N^*}\in\ell^2(\C).$ 
Before proceeding with the main results of the work, we introduce the following notations. For $n,j,k\in\N^*$ and $t\in [0,T]$ with $T>0$, we denote
\begin{equation}\begin{split}\label{definizioni}
&T^*:=\frac{\pi}{|B_{k,j}|},\\[5pt]
u_n(t):=&\frac{\cos\big((k^2-j^2)\pi^2t\big)}{n},\\[5pt]
b:=\iii B\iii_{(2)}^6\iii B\iii &\iii B\iii_3^{16}\max\big\{\iii B\iii,\iii B\iii_3\big\},\\[5pt]
E(j,k):=e^{\frac{6\iii B\iii_{(2)}}{|B_{j,k}|}}&|k^2-j^2|^5k^{24}\max\{j,k\}^{24}C_k^{-16}|B_{j,k}|^{-7},\\[5pt]
C':=\sup_{(l,m)\in \Lambda '}&\left\{\left|\sin\left(\pi\frac{|l^2-m^2|}{|k^2-j^2|}\right)\right|^{-1}\right\},\\
\end{split}\end{equation}
\begin{equation*}\begin{split}
\Lambda ':=\big\{&(l,m)\in (\N^*)^2 :\ \{l,m\}\cap\{j,k\}\neq\emptyset,\ |l^2-m^2|\leq\frac{3}{2}|k^2-j^2|,\\
&\ |l^2-m^2|\neq|k^2-j^2|,\ \la\phi_l,B\phi_{m}\ra\neq 0\big\}.
\end{split}\end{equation*}

\medskip
The following theorem represents the main result of the work, which ensure the global exact controllability between eigenfunctions. We underline that the control time is explicit and $u_n$ defines a dynamics steering the initial data to the target one up to a defined distance when $n$ is sufficiently large. 
\begin{teorema}\label{mainteo}
	Let $j,n\in\N^*$ and $k\in\N^*$ be such that $k\neq j$ and 
	\begin{equation}\label{diosbor}m^2-k^2\neq k^2-l^2,\ \ \ \ \  \  \ \ \ \ \ \ \ \forall m,l\in\N^*,\ m,l\neq k.\end{equation}
	Let $B$ satisfy Assumptions I. If $n\geq 2^{51}\pi^{19}b\ (1+C') E(j,k),$ then 
	$$\exists \theta\in\R\ \ \ \ :\ \ \ \big\|\G_{nT^*}^{u_n}\phi_j -e^{i\theta}\phi_k\big\|_{(3)}
	\leq 3C_k^2(16k^3\iii B\iii_3^2)^{-1}$$
	with $C_k$ from Assumptions I. Moreover, there exists $u\in L^2((0,\frac{4}{\pi}),\R)$ so that 
	$$\|u\|_{L^2((0,\frac{4}{\pi}),\R)}\leq \frac{C_k}{\iii B\iii_3^2 k^3},	\ \ \ \ \ \ \ \ \ \ \G_{\frac{4}{\pi}}^{u}\G_{nT^*}^{u_n}\phi_j=e^{i\theta}\phi_k.$$
\end{teorema}
\begin{proof}
	See Section $\ref{proofmainteo}$.\qedhere
\end{proof}

\medskip
Examples of $k\in\N^*$ satisfying the relation $(\ref{diosbor})$ are those numbers $k\leq 3$. However, Theorem $\ref{mainteo}$ can be generalized for every $k\in\N^*$ by defining, for every $\phi_j$ and $\phi_k$, a dynamics steering $\phi_j$ in $\phi_k$ and passing through $\phi_1$. 
In addition, the choice of $\{u_n\}_{n\in\N^*}$ can be replaced by other $\frac{2\pi}{|\lambda_k-\lambda_j|}-$periodic controls by refering to $\cite{chambrion2}$, which is used in the proof of the theorem.

\medskip
Theorem $\ref{mainteo}$ is not optimal and its purpose is to exhibit readable results for general $B$, $j$ and $k$.
For any specific choice of $B$, $j$ and $k$, it is possible to retrace the proof in order to obtain sharper bounds by using stronger intermediate estimates. We briefly treat the example of $B:\psi\mapsto x^2\psi$, $j=2$ and $k=1$ in Section $\ref{esempio}$. In addition, even though the phase appearing in the result is not particularly relevant from a physical point of view, it can be avoided by rotating the state of its phase (provided in $\cite{chambrion2}$).

\subsection{Well-posedness}\label{well}
As mentioned in the introduction, Beauchard and Laurent prove in $\cite{laurent}$ the well-posedness of the bilinear Schr\"odinger equation in $H^3_{(0)}$. The result is provided with $B$ a multiplication operator for a suitable function $\mu\in H^3((0,1),\R)$. We rephrase the result in the following proposition.

\needspace{3\baselineskip}
\begin{prop}{$\cite[Lemma\ 1; \ Proposition\ 2]{laurent}$}\label{laura} $\text{   }$

	\noindent
	{\bf 1)} Let the function $\widetilde f$ be so that $\widetilde f(s,\cdot)\in H_0^1\cap H^3$ for almost every $s\in [0,T]$ with $T >0$ and $\widetilde f\in L^2((0,T), H_0^1\cap H^3)$. The map $G:t\mapsto\int_0^te^{iAs} \widetilde f(s) ds$ belongs to $C^0([0,T], H^3_{(0)})$. Moreover, 
	$$\|G\|_{L^{\infty}((0,T),H^3_{(0)})}\leq c_1(T)\|\widetilde f\|_{L^2((0,T),H^3\cap H^1_{(0)})},$$ where the constant $c_1(T)$ is uniformly bounded with $T$ in bounded intervals.
	
	\vspace{5pt}

	\noindent
	{\bf 2)} Let $\mu\in H^3((0,1),\R)$, $T >0$, $\psi^0\in H^3_{(0)}$ and $u\in L^2((0,T),\R)$. There
	exists a unique mild solution of the ($\ref{mainx1}$) in
	$H^3_{(0)}$ when $B$ is a multiplication operator with respect to $\mu$, i.e. there exists $\psi\in C^0([0,T],H^3_{(0)})$ such that
	\begin{equation}\label{form}
	\psi(t)=e^{-iA t}\psi^0-i\int_0^t e^{-iA(t-s)}u(s)\mu\psi(s)ds,\ \ \ \ \ \forall t\in[0,T].
	\end{equation}
	Moreover, for every $R>0$, there exists $C=C(T,\mu,R)>0$ such that, for every $\psi^0\in H^3_{(0)}$, if $\|u\|_{L^2((0,T),\R)}<R ,$ then the solution satisfies
	$$\|\psi\|_{C^0([0,T],H^3_{(0)})}\leq C\|\psi^0\|_{(3)},\ \ \ \ \ \|\psi(t)\|_\Hi =\|\psi^0\|_\Hi\ \ \ \forall t\in[0,T].$$
\end{prop}

\begin{osss}\label{lauraa}
	The outcome of Proposition $\ref{laura}$ is not only valid for multiplication operators, but also for other suitable operators $B$. Indeed, the same proofs of $\cite[Lemma\ 1]{laurent}$ and $\cite[Proposition\ 2]{laurent}$ lead to the well-posedness of the $(\ref{mainx1})$ when $B$ is a bounded symmetric operator such that 
	$$B\in L(H^3_{(0)},H^3\cap H^1_0),\ \ \ \ \ \ B\in L(H^2_{(0)}),$$
	which are verified if $B$ satisfies Assumptions I, thanks to $\cite[Remark\ 1.1]{mio1}$.
\end{osss}

\subsection{Scheme of the work}
In Section $\ref{globasec}$, Proposition $\ref{lcl}$ ensures the local exact controllability in $H^3_{(0)}$ and we exhibit a neighborhood where it is verified in Proposition $\ref{lclneigh}$.
We prove Theorem $\ref{mainteo}$ in Section $\ref{proofmainteo},$ while we apply the main results to a physical system in Section $\ref{esempio}$.
In Section $\ref{conclusion},$ we comment the outcomes of Theorem $\ref{mainteo}$. 
We provide some intermediate results in Appendix $\ref{pallose}$, while in Appendix $\ref{appmome1}$, we expose some tools required in the work.

\section{Local exact controllability in $H^3_{(0)}$}\label{globasec}

Let us provide a brief proof of the local exact controllability in $H^3_{(0)}$ by rephrasing the existing results of local exact controllability as $\cite{laurent}$, $\cite{morgane1}$ and $\cite{morganerse2}$. Our purpose is to introduce the tools that we use in the proof of Theorem $\ref{mainteo}$. For $\psi\in H^3_{(0)}$ and $\epsilon>0$, we define
\begin{equation}\label{dindo}\begin{split}
\widetilde B_{H^3_{(0)}}(\psi, \epsilon):= \big\{&\widetilde\psi\in  H^3_{(0)}\big|\ \|\widetilde\psi\|=\|\psi\| ,\ \|\widetilde\psi-\psi\|_{(3)}<\epsilon\big\}.
\end{split}\end{equation}

\begin{prop}\label{lcl}
	Let $B$ satisfy Assumptions I. For every $l\in\N^*$ such that
	\begin{equation}\label{diosbor1}m^2-l^2\neq l^2-n^2,\ \ \ \ \  \  \ \ \ \ \ \ \ \forall m,n\in\N^*,\ m,n\neq l,\end{equation}
	there exist $T>0$ and $\epsilon>0$ such that, for every $\psi\in \widetilde B_{H^3_{(0)}}(\phi_l(T), \epsilon)$, there exists a control function $u\in L^2((0,T),\R)$ such that $\psi= \G^u_T\phi_l.$
\end{prop} 
\begin{proof}
	Let $S$ be the unit sphere in $\Hi$. Proposition $\ref{lcl}$ can be proved by ensuring the surjectivity, for $T>0$ large enough, of the map $$\G_{T}^{(\cdot)}\phi_l: L^2((0,T),\R)\rightarrow \widetilde B_{H^3_{(0)}}(\phi_l(T), \epsilon)\subset H^3_{(0)}\cap S$$ with suitable $\epsilon>0$. We prove this property and that the preimage of $\widetilde B_{H^3_{(0)}}(\phi_l(T), \epsilon)$ is a neighborhood of $u_0=0$ in $L^2((0,T),\R)$.
	Let $$\G_{t}^u\phi_l=\sum_{k\in\N^*}{\phi_k(t)}\la \phi_k(t),\G_{t}^u\phi_l\ra.$$ Let $\alpha_l(\cdot)=\{\alpha_{k,l}(\cdot)\}_{k\in\N^*}$ be such that $\alpha_{k,l}(\cdot)=\la \phi_k(T), \G_{T}^{(\cdot)}\phi_l\ra$ for $k\in\N^*$ and
	$$\alpha_l:L^2((0,T),\R)\longrightarrow Q:=\{{\bf x}\in h^3(\C)\ |\ \|{\bf x}\|_{\ell^2}=1\}.$$ Let $\delta_l:=\{\delta_{k,l}\}_{k\in\N^*}$. The statement follows from the surjectivity of the map $\alpha_l:L^2((0,T),\R)\longrightarrow Q_\epsilon:=\{{\bf x}\in h^3(\C)\ |\ \|{\bf x}\|_{\ell^2}=1,\  \|{\bf x}-\delta_l\|_{(3)}\leq \epsilon\}$. We use the Generalized Inverse Function Theorem (\cite[Theorem\  1;\ p.\ 240]{Inv}) and we study the surjectivity of the Fréchet derivative of $\alpha_l$
	$$\gamma_l:=d_u\alpha_l(0):L^2((0,T),\R)\longrightarrow T_{\updelta_l}Q=\{\{x_k\}_{k\in\N^*}\in h^3(\C)\ |\ ix_l\in\R\},$$
	$$\gamma_l(v):=\{\gamma_{k,l}(v)\}_{k\in\N^*},\ \ \ \ \gamma_{k,l}(v):=-i\int_{0}^Tv(s)e^{i(\lambda_k-\lambda_l)s}ds B_{k,l},\ \ \  \forall k\in\N^*.$$ To this end, we show there exists $T>0$ so that, for every $\{x_k\}_{k\in\N^*}\in T_{\updelta_l}Q$,
	\begin{equation}\label{mome}\begin{split}
	\exists u\in L^2((0,T),\R)\ \ \ :\ \  \ \frac{x_{k}}{B_{k,l}}=-i\int_{0}^Tu(s)e^{i(\lambda_k-\lambda_l)s}ds,\ \ \ \ \ \ \forall k\in\N^*.\end{split}\end{equation}
	The solvability of the moment problem $(\ref{mome})$ is equivalent to the surjectivity of $\gamma_l$. As $B$ is symmetric, there holds $B_{l,l}\in\R$ and $i\big(x_{l}/B_{l,l}\big)\in\R$. Moreover, $\{x_{k}/B_{k,l}\}_{{k\in\N^*}}\in\ell^2(\C)$ since $\{x_k\}_{k\in\N^*}\in h^3(\C)$ and thanks to Assumptions I. The solvability of $(\ref{mome})$ follows from Lemma $\ref{solve112}$ for $T$ large enough, since 
	$$\{ix_{k}/B_{k,l}\}_{{k\in\N^*}}\in \{\{a_k\}_{k\in\N^*}\in\ell^2(\C):\ a_{l}\in\R\}.$$ For $X$ from Lemma $\ref{solve112}$, the map $\gamma_l:X\longrightarrow T_{\updelta_l}Q$ is an homeomorphism and $\gamma_l:L^2((0,T),\R)\rightarrow T_{\updelta_l}Q$ is surjective in $T_{\updelta_l}Q$ for $T$ large enough. The proof is achieved as the map $\alpha_l$ is surjective in $Q_\epsilon$ for $\epsilon>0$ small enough.	\qedhere
\end{proof}

\subsection{Local exact controllability  in an explicit neighborhood}\label{stimiamo}

Let $C_l$ and $\widetilde B_{H^3_{(0)}}(\cdot,\cdot)$ be respectively defined in Assumptions I and $(\ref{dindo})$. 
The following proposition ensures the local exact controllablity in an explicit neighborood of $H^3_{(0)}$ for a specific time. The result leads to Theorem $\ref{mainteo}$.

\begin{prop}\label{lclneigh}
	Let $B$ satisfy Assumptions I and $l\in\N^*$ be such that
	\begin{equation}\label{diosbor112}m^2-l^2\neq l^2-n^2,\ \ \ \ \  \  \ \ \ \ \ \ \ \forall m,n\in\N^*,\ m,n\neq l.\end{equation}
	For every $\psi\in \widetilde B_{H^3_{(0)}}\big(\phi_l\big(\frac{4}{\pi}\big),  \frac{3C_l^2}{16l^3 \iii B\iii_3^2}\big),$ there exists $u\in L^2((0,\frac{4}{\pi}),\R)$ so that 
	$$\psi= \G^u_{\frac{4}{\pi}}\phi_l.$$
\end{prop} 
\begin{proof}
	Let us define the following notations
	$$\iii \cdot\iii_{L(L^2((0,T),\R),H^3_{(0)})}=\iii \cdot \iii_{(L_t^2,H_x^3)},\ \ \ \iii \cdot\iii_{L(H^3_{(0)},L^2((0,T),\R))}=\iii \cdot \iii_{(H_x^3,L_t^2)},$$
	$$\|\cdot\|_{L^{\infty}((0,T),H^3_{(0)})}=\|\cdot\|_{L_t^\infty H_x^3},\ \ \  \|\cdot\|_{L^2((0,T),\R)}=\|\cdot\|_{2}.$$
	Let $T>\frac{2\pi}{\Gi}$ for $\Gi=\pi^2$ and $X$ be defined in the proof of Lemma $\ref{solve112}$. In the proof of Proposition $\ref{lcl}$, 
	we ensure that $\gamma_l:X\longrightarrow T_{\updelta_l}Q$ is an homeomorphism and that $\alpha_l:L^2((0,T),\R)\rightarrow H^3_{(0)}$ is locally surjective. Let $$A_l(\cdot)=\sum_{j\in\N^*}\phi_j(T)\alpha_{j,l}(\cdot)=\G^{(\cdot)}_T\phi_l,\ \ \ \ \ \  F_l(\cdot):=\sum_{j\in\N^*}\phi_j(T)\gamma_{j,l}(\cdot),$$ $$A_l(\cdot):V\subseteq L^2((0,T),\R)\rightarrow \{\psi\in H^3_{(0)}\ :\ \|\psi\|_{\Hi}=1\},$$ $$F_l(\cdot): X\rightarrow \{\psi\in H^3_{(0)}:\ i\la\phi_l(T),\psi \ra\in \R\}.$$
	By definition, the map $F_l$ is an homeomorphism and $A_l$ is locally surjective.
	We use $\cite[Lemma\  2.3;\ p.\ 42]{nabil}$ and we estimate a neighborhood where $A_l$ is surjective as $X$ and $\{\psi\in H^3_{(0)}:\ i\la\phi_l(T),\psi \ra\in \R\}$ are Banach spaces. In particular, we compute a constant $M>0$ such that
	\begin{equation}\label{11}\|F_l(u)-F_l(v)\|_{(3)}\geq M\|u-v\|_{2},\ \ \ \ \ \ \ \forall u,v\in X.\end{equation}
	
	\noindent
	Fixed $T>0$ large enough, we provide $U\subset X$ and $M_1<M$ such that 
	\begin{equation}\label{11dd}\|(A_l-F_l)(u)-(A_l-F_l)(v)\|_{(3)}\leq M_1\|u-v\|_{2},\ \ \ \ \ \ \ \forall u,v\in U.\end{equation}

	\noindent
	When $(\ref{11})$ and $(\ref{11dd})$ are satisfied, $\cite[Lemma\  2.3;\ p.\ 42]{nabil}$ ensures that $A_l:U\longrightarrow A_l(U)$ is an homeomorphism. In addition, the proof of the cited lemma implies that, if $U\supset\{ u\in X:\|u\|_{2}\leq r\}$ with $r>0$, then
	$$A_l(U)\supset\{\psi\in H^3_{(0)}:\|\psi-\phi_l(T)\|_{(3)}\leq r(M-M_1) \}.$$
	
	\medskip
	\needspace{3\baselineskip}
	
	\noindent
	{\bf 1)} We compute $M>0$ such that $(\ref{11})$ is verified. As $F_l$ is an homeomorphism, for every $\psi\in H^3_{(0)}$, there exist $T>0$ and $u\in X$ such that, for every $j\in\N^*$
	\begin{equation}\label{dioporco}\la\phi_j(T),\psi\ra=\la\phi_j(T),F_l(u)\ra=\gamma_{j,l}(u),\ \ \ \ \ F_l^{-1}(\psi)=u. \end{equation} 
	Now, $\frac{\la\phi_k(T),\psi\ra}{B_{k,l}}=\frac{\gamma_{k,l}(u)}{B_{k,l}}= -i\int_{0}^Tu(s)e^{i(\lambda_k-\lambda_l)s}ds$ for every $k\in\N^*$, thanks to the validity of $(\ref{mome})$.
	Thus, from the proof of Lemma $\ref{solve112}$, we have $$u(t)= \frac{\la\phi_l(T),\psi\ra}{B_{l,l}}v_{l}(t)+2\sum_{k\in\N^*\setminus{l}} \Re\Big(\frac{\la\phi_k(T),\psi\ra}{B_{k,l}} {v_k}(t)\Big)$$ for $\{v_k\}_{k\in\Z}$ the unique biorthogonal family to $\{e^{i\lambda_k (\cdot)}\}_{k\in\Z}$. In addition, from $(\ref{weirdos})$, there exists $\widetilde C(T)>0$ such that $\|u\|_{2}^2\leq \widetilde C(T)^2\sum_{j=1}^{\infty}\big|\frac{\gamma_{j,l}(u)}{B_{j,l}}\big|^2$ and
	\begin{equation*}
	\begin{split}
	\| F_l^{-1}(\psi)\|_{2}^2&=\|u\|_{2}^2 \leq \frac{\widetilde C(T)^2}{C_l^2}\sum_{j=1}^{\infty}|j^3\gamma_{j,l}(u)|^2\leq \frac{\widetilde C(T)^2}{C_l^2}\|\psi\|_{(3)}^2.\\
	\end{split}
	\end{equation*}
	In conclusion, we set $M={C_l}/{\widetilde C(T)}$ since, for every $\psi,\ffi\in H^3_{(0)}$, there exist $v,w\in X$ such that $\psi=F_l(v),$ $\ffi=F_l(w)$ and
	$$\|v-w\|_{2}\leq \|F_l^{-1}(\psi-\ffi)\|_{2}\leq \frac{\widetilde C(T)}{C_l}\|\psi-\ffi\|_{(3)}.$$

	\medskip

	\needspace{3\baselineskip}
	
	\noindent
	{\bf 2)} We suppose $\iii B\iii_{3}=1$. For $u\in X$, from the Duhamel's formula,
	\begin{equation}\label{der}\begin{split}
	&\G_T^u\phi_l=e^{-i\lambda_l T}\phi_l-i\int_0^T e^{-iA(T-s)}u(s) B \G_s^u\phi_l\\
	&=e^{-i\lambda_l T}\phi_l-i\int_0^T e^{-iA(T-s)}u(s) B e^{-i\lambda_l s}\phi_l ds\\
	&-\int_0^T e^{-iA(T-s)}u(s) B\Big( \int_0^{s} e^{-iA(s-{\tau})}u({\tau}) B\G_{\tau}^u\phi_ld{\tau}\Big)ds
	\end{split}\end{equation}
	Let $H_l(u):=-\int_0^T e^{-iA(T-s)}u(s) B \int_0^{s} e^{-iA(s-{\tau})}u({\tau}) B\G_{\tau}^u\phi_ld{\tau}ds$. From $(\ref{der})$,
	\begin{equation*}\begin{split}
	&A_l(u)=\G_T^u\phi_l=e^{-i\lambda_l T}\phi_l+F_l(u)+H_l(u).\\
	\end{split}\end{equation*}
	We recall that we aim to exhibit a ball $U\subset X$ with center $u=0$ such that, for every $u\in U$, the map $A_l:u\in U\mapsto \G_T^u\phi_l\in A_l(U)$ is an homeomorphism by using $\cite[Lemma\  2.3;\ p.\ 42]{nabil}$. 
	To this purpose, we construct $U$ such that there exists $M_1>0$ satisfying $(\ref{11dd})$ and $M_1\leq M/{2}.$ We notice that
	$$\|(A_l-F_l)(u)-(A_l-F_l)(v)\|_{(3)}= \|H_l(u)-H_l(v)\|_{(3)},\ \ \ \ \forall u,v\in L^2((0,T),\R),$$ 
	\begin{equation}\begin{split}\label{distanza}
	H_l(u)-H_l(v)&=-\int_0^T e^{-iA(T-s)}(u(s)-v(s)) B\Big( \int_0^{s} e^{-iA(s-{\tau})}u({\tau}) B\G_{\tau}^u\phi_ld{\tau}\Big)ds\\
	&-\int_0^T e^{-iA(T-s)}v(s) B\Big( \int_0^{s} e^{-iA(s-{\tau})}(u({\tau})-v({\tau})) B\G_{\tau}^u\phi_ld{\tau}\Big)ds\\
	&-\int_0^T e^{-iA(T-s)}v(s) B\Big( \int_0^{s} e^{-iA(s-{\tau})}v({\tau}) B(\G_{\tau}^u\phi_l-\G_{\tau}^v\phi_l)d{\tau}\Big)ds.\\
	\end{split}\end{equation}
	Thanks to Proposition $\ref{laura}$ and Remark $\ref{lauraa}$, there exists a constant $C(T)>0$ such that, for every $\psi\in L^\infty((0,T),H^3_{(0)})$ and $u\in L^2((0,T),\R)$,
	\begin{equation}\label{russo}\left\|\int_0^T e^{-iA(T-s)}u(s) B\psi ds
	\right\|_{(3)}\leq C(T)\|u\|_{2}\iii B\iii_3\|\psi \|_{L_t^\infty H_x^3}.\end{equation}
	As we assumed $\iii B\iii_3=1$, we have
		\begin{equation*}\begin{split}
	&\|\G_t^v\phi_l-\G_t^u\phi_l\|_{L_t^\infty H_x^3}\leq\Big\|\int_0^te^{-iA(t-s)}B(v\G_t^v\phi_l-u\G_t^u\phi_l)\Big\|_{L_t^\infty H_x^3}\\
	&\leq C(T)\iii B\iii_{3} \|v\G_t^v\phi_l-u\G_t^u\phi_l\|_{L_t^\infty H_x^3}\\
	&\leq 
	C(T)\|v-u\|_{2}\|\G_t^u\phi_l\|_{L_t^\infty H_x^3}
	+C(T)\|v\|_{2}\|\G_t^v\phi_l-\G_t^u\phi_l\|_{L_t^\infty H_x^3}.
	\end{split}\end{equation*}
	Let $\mu >1$. If $U=\{u\in X:\ \|u\|_{2}\leq{(\mu C(T))}^{-1}\},$ then 
	\begin{equation}\label{morghy}\begin{split}
	\|\G_t^v\phi_l-\G_t^u\phi_l\|_{L_t^\infty H_x^3}&
	\leq{\frac{\mu C(T) }{\mu -1}\|v-u\|_{2}\|\G_t^u\phi_l\|_{L_t^\infty H_x^3}}\\
	\end{split}\end{equation}
	for every $u,v\in U$. From $(\ref{distanza})$, when $\|u\|_{2},\|v\|_{2}\leq{(\mu C(T))}^{-1}$,
	\begin{equation*}\begin{split}
	&\|H_l(u)-H_l(v)\|_{{(3)}}
	\leq C(T)^2\big(\|v-u\|_{2}(\|u\|_{2}+\|v\|_{2})\|\G_t^u\phi_l\|_{L_t^\infty H_x^3}\\
	&+C(T)^2\|v\|^2_{2}\|\G_t^v\phi_l-\G_t^u\phi_l\|_{L_t^\infty H_x^3}\\
	&\leq  2\mu^{-1}C(T)\|v-u\|_{2}\|\G_t^u\phi_l\|_{L_t^\infty H_x^3}+\mu^{-2}\|\G_t^v\phi_l-\G_t^u\phi_l\|_{L_t^\infty H_x^3}
	\end{split}\end{equation*}
	and, thanks to $(\ref{morghy})$, we have
	\begin{equation*}\begin{split}
	&\|H_l(u)-H_l(v)\|_{{(3)}}\leq \frac{(2\mu -1)}{(\mu -1)\mu }C(T)\|v-u\|_{2}\|\G_t^u\phi_l\|_{L_t^\infty H_x^3}.\\
	\end{split}\end{equation*}
	Thanks to the relation $(\ref{russo})$ and to the Duhamel's formula,
	$$\|\G_T^u\phi_l\|_{L_t^\infty H^3_x}\leq \|\phi_l\|_{(3)}+C(T)\|u\|_{2}\iii B\iii_{3}\|\G_T^u\phi_l\|_{L_t^\infty H^3_x}.$$
	We obtain $\|\G_T^u\phi_l\|_{L_t^\infty H^3_x}\leq \frac{\|\phi_l\|_{(3)}}{1-C(T)\|u\|_{2}\iii B\iii_{3}}\leq \frac{\mu l^3}{\mu -1}$ and, for every $u,v\in U$,
	\begin{equation*}\begin{split}
	&\Longrightarrow\ \ \ \ \ \|H_l(u)-H_l(v)\|_{{(3)}}\leq \frac{2\mu -1}{(\mu -1)^2}l^3C(T)\|v-u\|_{2}.\\
	\end{split}\end{equation*}
	To apply $\cite[Lemma\  2.3;\ p.\ 42]{nabil}$, we set $M_1=\frac{2\mu -1}{(\mu -1)^2}l^3C(T)$ and we estimate $\mu $ such that $M_1\leq M/{2}.$ 
	The last inequality is true when
	\begin{equation}\label{cazz1}\mu \geq a_l+\sqrt{a_l(a_l+1)}+1,\ \ \ \ \ \ \ \ a_l:={2C(T)\widetilde C(T)l^3}{C_l^{-1}}.\end{equation} 
		Let $T=\frac{4}{\pi}$. We notice that $a_l\leq \frac{6}{5}\widetilde a_l$ for $\widetilde a_l:={l^3}/{C_l}$ as
		$$C\Big(\frac{4}{\pi}\Big)\widetilde C\Big(\frac{4}{\pi}\Big)\leq \frac{3}{5}.$$
		We study the constants $C_1,C_2$ appearing in $(\ref{cris})$ that is valid thanks to the Ingham's Theorem ($\cite[Theorem \ 4.3]{Ing}$).  Let $|I'|:=\frac{\Gi}{\pi}T=4$, $\beta=\frac{\pi^2}{4}$, $G(0)=\frac{\pi}{2},$ $I_0=[-1,+1],$ $m=\big({|I'|}{|I_0|^{-1}}\big)=2,$ $\alpha=4R^2$, $\widehat G(0)=\frac{(R^2-1)\pi}{2}$ and $R=\frac{|I'|}{2}=2.$
		By substituting these parameters in the proof of Ingham's Theorem $\cite[pp.\ 62-65]{Ing}$), we obtain $$C_2^2=\frac{2m\pi G(0)\pi}{\beta \Gi}=\frac{8}{\pi},\ \ \ \ \ \ C_1^2=\frac{2\pi \widehat G(0)\pi}{\alpha \Gi}=\frac{3\pi}{16}.$$
		The proof of Proposition $\ref{laura}$ and $(\ref{weirdos})$, imply
		\begin{equation}\label{orcod}C(T)=C\big({4}/{\pi}\big)=3\pi^{-3}\max\big\{\sqrt{2}C_2,\sqrt{{4}/{\pi}}\big\}=3\pi^{-3}\sqrt{2}C_2.\end{equation}
		In addition, we have $\widetilde C\big(\frac{4}{\pi}\big)=2C_1^{-1}$, $C\big(\frac{4}{\pi}\big)\widetilde C\big(\frac{4}{\pi}\big)\leq \frac{3}{5}$ and $a_l\leq \frac{6}{5}\widetilde a_l$.	
	Now, $$C_l\leq |\la\phi_1,B\phi_l\ra| \leq\iii B\iii=1,\ \ \ \ \ \Longrightarrow\ \ \ \ \  \ \widetilde a_l>1.$$ We need to define $\mu$ such that $(\ref{cazz1})$ is verified and we notice that
	\begin{equation*}\begin{split} &a_l+\sqrt{a_l(a_l+1)}+1\leq \Big(\frac{6}{5}\widetilde a_l+\sqrt{\frac{6}{5}\widetilde a_l\Big(\frac{6}{5}\widetilde a_l+1\Big)}+1\Big)\leq \frac{22}{5}\widetilde a_l=\frac{22}{5}\frac{l^3}{C_l}.\end{split}\end{equation*}
	If we set $\mu=\frac{22}{5}\frac{l^3}{C_l}$, then $\mu\geq a_l+\sqrt{a_l(a_l+1)}+1$ as required in $(\ref{cazz1})$ and $$U=\big\{u\in X:\ \|u\|_{2}\leq{(\mu C({4}/{\pi}))}^{-1}\big\}.$$ Since $M_1\leq M/{2}$, we have the validity of $\cite[Lemma\  2.3;\ p.\ 42]{nabil}$, which implies that $A_l:U\subseteq L^2((0,4/\pi),\R)\rightarrow A(U)\subseteq H^3_{(0)}$ is an homeomorphism.
	
	\medskip
	
	\needspace{3\baselineskip}
	
	\noindent
	{\bf 3)} We show a neighborhood of $\phi_l$ in $H^3_{(0)}$ contained in $A_l(U).$ Let
	$$B_{X}(x,r):=\{\widetilde x\in X\ \big| \  \|\widetilde x-x\|_{L^2((0,\frac{4}{\pi}),\R)}\leq r\}.$$ We notice that $\mu C\Big(\frac{4}{\pi}\Big)<\frac{l^3}{C_l}$ and we set $\widetilde U=B_{X}\left(0,\frac{C_l}{l^3}\right)\subset U.$ From the proof of $\cite[Lemma\  2.3;\ p.\ 42]{nabil}$, we know that $A_l(U)$ contains a ball of center $A_l(0)=\phi_l\big(\frac{4}{\pi}\big)$ and radius $(M-M_1)\frac{C_l}{l^3}$.
	As $M_1\leq M/{2}$, we have 
	$$M-M_1\geq\frac{1}{2}M\geq\frac{{C_l}}{2\widetilde C\Big(\frac{4}{\pi}\Big)}=\frac{{3}{C_l}}{16},$$
	\begin{equation*}\begin{split}
	&A_l\left(B_{X}\left(0,\frac{C_l}{l^3}\right)\right)\supseteq \widetilde B_{H^3_{(0)}}\left(A_l(0),(M-M_1)\frac{C_l}{l^3}\right)\supseteq \widetilde B_{H^3_{(0)}}\left(\phi_l\Big(\frac{4}{\pi}\Big),\frac{3C_l^2}{16 l^3}\right).
	\end{split}
	\end{equation*}
	In the second part of the proof, we suppose $\iii B\iii_{3}=1$, but we can generalize the result for $\iii B\iii_3\neq 1$ thanks to the identity $A+uB=A+u\iii B\iii_{3} \frac{B}{\iii B\iii_{3} }.$ We consider the operator $\frac{B}{\iii B\iii_{3}}$ and the control $u\iii B\iii_{3}$, while we substitute $C_l$ with $C_l \iii B\iii_3^{-1}$ (from Assumptions I). In addition, if $\iii B\iii_{3}u\in B_{X}\left(0,\frac{C_l}{l^3\iii B\iii_3}\right)$, then $u\in B_{X}\left(0,\frac{C_l}{l^3\iii B\iii_3^2}\right).$	In conclusion,
	$$\forall\psi\in \widetilde B_{H^3_{(0)}}\left(\phi_l\Big(\frac{4}{\pi}\Big),\frac{ 3C_l^2}{16 l^3\iii B\iii_3^2}\right),\ \exists\ u\in B_{X}\left(0,\frac{C_l}{l^3\iii B\iii_3^2}\right): \ A_l(u)=\psi.\qedhere$$
\end{proof}

\section{Proof of Theorem\ $\ref{mainteo}$}\label{proofmainteo}
Let $T^*$, $u_n$, $b$, $E(j,k)$ and $C '$ be defined in $(\ref{definizioni})$.
	The proof follows from Proposition $\ref{lclneigh}$ and Proposition $\ref{approssimaH3T}$. From Proposition $\ref{approssimaH3T}$, we have  
	\begin{equation*}\begin{split}
	R''_n:=\|\G_{nT^*}^{u_n}\phi_j-e^{i\theta}\phi_k\|_{(3)}^8\leq&\frac{2^{20} 3^{2}\pi^{24}e^{\frac{6\iii B\iii_{(2)}}{B_{j,k}}}|k^2-j^2|^5
		\max\{j,k\}^{24}}{|B_{j,k}|^{7}n}\cdot\\
	&\frac{(1+C') \iii B\iii_{(2)}^6\iii B\iii \max\{\iii B\iii,\iii B\iii_3\}}{|B_{j,k}|^{7}n}.\\
	\end{split}\end{equation*}
	We know $\lim_{n\rightarrow\infty}R''_n=0$. Let us provide an explicit $n^*$ so that
	\begin{equation}\label{bongobongo}\G_{n^*T^*}^{u_{n^*}}\phi_j\in \widetilde B_{H^3_{(0)}}\Bigg(e^{i\theta}\phi_k,\frac{3C_k^2}{2^4k^3 \iii B\iii_3^2}\Bigg),\ \ \ \  \ R''_{n^*}\leq\frac{3^8C_k^{16}}{2^{32}k^{24} \iii B\iii_3^{16} }.\end{equation}
	For $0\leq s < 3$ and $j,k\in\N^*$, $\iii B\iii_{(s)}\geq C_k$ and $\iii B\iii_{(s)}\geq |B_{j,k}|$. If
	\begin{equation*}\begin{split}
	n^*\geq\frac{2^{51}\pi^{19}e^{\frac{6\iii B\iii_{(2)}}{B_{j,k}}}
		b(1+C')|k^2-j^2|^5k^{24}\max\{j,k\}^{24}}{C_k^{16}|B_{j,k}|^{7}}\\
	\end{split}\end{equation*}
	then $(\ref{bongobongo})$ is valid with $b$ from $(\ref{definizioni})$. Thanks to Proposition $\ref{lclneigh}$ and to the time reversibility of the $(\ref{mainx1})$ (see $\cite[Section\ 1.3]{mio1}$), we obtain
	\begin{equation}\label{oo} 
	\exists u\in L^2\Big(\Big(0,\frac{4}{\pi}\Big),\R\Big)\ \ \ \  :\ \ \  \ \G_{\frac{4}{\pi}}^{u}\G_{n^*T^*}^{u_{n^*}}\phi_j=e^{i\theta}\phi_k.\end{equation}

\section{Application of the main result}\label{esempio}
In the current section, we briefly propose a possible application of Theorem $\ref{mainteo}.$
Let us consider an electron trapped in a one-dimensional guide of length $\sim 10^{-3}\ m$ and represented by the quantum state $\uppsi$. We suppose that the electron is subjected to an external time-depending electromagnetic field $\mathcal{V}(\tau)$ with $\tau\in [0, \text{T}]$ and $\text{T}$ a positive time. Let $m_e\sim 10^{-30}\  Kg$ be the mass of the electron and $\hbar\sim 10^{-34}\ \frac{m^2\  Kg}{s}$ with $\hbar$ the reduced Planck constant. The dynamics of $\uppsi$ is modeled by the Schr\"odinger equation
\begin{equation}\label{s}\begin{split}
i \hbar\ \frac{d}{d\tau}\uppsi(\tau)=-\frac{\hbar^2}{m_e}\ \frac{d^2}{d\text{x}}\uppsi(\tau)+\mathcal{V}(\tau)\uppsi(\tau),\ \ \ \ \ \ \tau\in (0,\text{T}).\end{split}\end{equation}
We substitute $x:=\text{x}\   10^{3}\  {m}^{-1}$, $t:=\tau \   10^{2}\ {s}^{-1}$ and $\psi(t,x):=\uppsi(\tau,\text{x}).$ Now, $$V(t):=\mathcal{V}(\tau)\ 10^{32}\ \frac{s^2}{m^2\  Kg},\ \ \ \ \ (t,x)\in (0,T)\times(0,1),\ \ \ \ T:=\text{T} \   10^{2}\ {s}^{-1}$$ are dimensionless (without unit of measurement) and $(\ref{s})$ corresponds to
$$i\ \frac{d}{d t}\psi(t)=-\frac{d^2}{d x^2}\psi(t)+ V(t)\psi(t),\ \ \ \ \ t\in(0,T).$$
If the potential $V(t,x)$ is equal to $u(t)x^2$ , then we obtain the $(\ref{mainx1})$
$$i\dd_{ t}\psi( t,x)=A\psi(t,x)+u( t)x^2\psi(t,x).$$

\smallskip
We point out that the last equation can be used to model the dynamics of an electron subjected to two external fields. The first one forces its behaviour to a quantum harmonic oscillator with time dependent intensity. The second field instead traps the electron in a potential well.
\medskip

We exhibit $u$ driving the state of the particle from the first excited state to the ground state. For this reason, we retrace the proof of the first point of Theorem $\ref{mainteo}$ with $B:\psi\mapsto x^2\psi$. Let $\phi_1$ and $\phi_2$ be eigenstates of $A$. We define a control function driving $\phi_2$ in $\phi_1$. We notice that $\la\phi_j, x^2\phi_k\ra={2}\int_0^1x^2\sin(\pi jx)\sin(\pi kx)dx$ and
Assumptions I are satisfied since
\begin{equation*}\begin{split}
|\la\phi_j, x^2\phi_k\ra|&=
\Big|\frac{(-1)^{j-k}}{(j-k)^2\pi^2}-\frac{(-1)^{j+k}}{(j+k)^2\pi^2}\Big|=
\frac{4jk}{(j^2-k^2)^2\pi^2},\ \ \ \ \ j\neq k,\\
|\la\phi_k, x^2\phi_k\ra|&=\Big|\frac{1}{3}-
\frac{1}{2k^2\pi^2}\Big|,\ \ \ \ \ \ \ \ \ \ \ \ \ \ \ \ \ \ k\in\N^*.\\
\end{split}
\end{equation*}

\noindent
For $\psi\in H^3_{(0)}$, we have $x^2\psi\in H^3\cap H^1_0$, $\|x\psi\|\leq\|\psi\|,$ $\|x^2\psi\|\leq\|\psi\|$ and 
$$\|\dd_x\psi\|^2=\||A|^\frac{1}{2}\psi\|^2=\sum_{k\in\N^*}|\pi k \la\phi_k,\psi\ra|^2\leq\sum_{k\in\N^*}|(\pi k)^2\la\phi_k,\psi\ra|^2=\|\dd_x^2\psi\|^2.$$ From the Poincar\'e inequality, $\|\psi\|\leq\frac{1}{\pi}\|\dd_x\psi\|$ and $\|\dd_x^2\psi\|\leq\frac{1}{\pi}\|\dd_x^3\psi\|.$ Thus,
\begin{equation*}\begin{split}
\|\dd_x(x^2\psi)\|&\leq\|2x\psi\|+\|x^2\dd_x\psi\|\leq\frac{2}{\pi^2}\|\dd_x^3\psi\|+\frac{1}{\pi}\|\dd_x^3\psi\|\leq\frac{2+\pi}{\pi^2}\||A|^\frac{3}{2}\psi\|,\\
\|\dd_x^2(x^2\psi)\|&\leq\|2\psi\|+\|4x\dd_x\psi\|+\|x^2\dd_x^2\psi\|\leq\frac{2+5\pi}{\pi^2} \||A|^\frac{3}{2}\psi\|,\\
\|\dd_x^3(x^2\psi)\|&\leq\|6\dd_x\psi\|+\|6x\dd_x^2\psi\|+
\|x^2\dd_x^3\psi\|\leq
\frac{12+\pi}{\pi}\||A|^\frac{3}{2}\psi\|,\\
\end{split}
\end{equation*}
\begin{equation*}\begin{split}
&\iii B\iii_3=\sup_{\underset{\|\psi\|_{(3)}\leq 1}{\psi\in H^3_{(0)}}}\sqrt{\|\dd_x (x^2\psi)\|^2+\|\dd_x^2 (x^2\psi)\|^2+\|\dd_x^3 (x^2\psi)\|^2}\leq\ 5.2\ .\\
\end{split}
\end{equation*}
Similarly, $\iii B\iii_{(2)}\leq 1.64 $ and $\iii B\iii\leq 1$. Moreover, $C'=0$ and
$$|B_{1,1}|=C_1=\frac{2\pi-3}{6\pi^2},\ \ \ \  |B_{1,2}|=C_2=\frac{8}{9\pi^2},\ \ \ \ \ I=\frac{4}{3\pi^2}.$$
If we retrace the proof of Proposition $\ref{lclneigh}$ by substituting the previous constants, then we see that the local exact controllability is verified in
$\widetilde B_{H^3_{(0)}}(\phi_1,\ 2.14\  10^{-5})$. Let $T=\frac{2}{3\pi}$, $u(t)=\cos(3\pi^2 t)$, $T^*=\frac{9\pi^3}{8}$, $K=\frac{9\pi^2}{4}$. By repeating the proof of Theorem $\ref{mainteo}$ and Proposition $\ref{approssimaH3T}$, for $u_n:=\frac{u}{n}$,
\begin{equation*}\begin{split}
&\exists \theta\in\R\ \ \ \ \ :\ \ \ \ \|e^{i\theta} \phi_1-\G_{nT^*}^{u_n} \phi_2\|^8_{(3)}\leq { 10^{80}}{n^{-1}}.\\
\end{split}\end{equation*}
In the neighborhood $\widetilde B_{H^3_{(0)}}\left(\phi_1,\  2.14\  10^{-5}\right),$ the local exact controllability is verified, while the first point of Theorem $\ref{mainteo}$ holds for $n=2.3\  10^{117}$. Let
$$u(t)=(2.3\  10^{117})^{-1}\cos(3\pi^2t), \ \ \ \ \ T=(2.3\  10^{117})\frac{9\pi^3}{8}.$$
There exists $ \theta\in\R$ such thats $\big\|e^{i\theta}\phi_1-\G_{T}^{u}\phi_2\big\|_{(3)}\leq 2.14\  10^{-5}.$
In addition, 
$$\exists \tilde u\in L^2((0,\frac{4}{\pi}),\R)\ \ \ \ :\ \ \ \ \G_{T}^{u}\G_{\frac{4}{\pi}}^{\tilde u}\phi_2=e^{i\theta}\phi_1.$$
In conclusion, the dynamics of $(\ref{s})$ drives the state of the electron from the first excited state to the ground state in a time $\text{T}\sim 10^{116}\ s.$

\section{Conclusion}\label{conclusion}

The results provided in the work represent a contribution for the application of the control theory to the physical experimentation for systems modeled by the bilinear Schr\"odinger equation. 
Given any couple of bounded states, we provided controls and times such that the dynamics of the $(\ref{mainx1})$ drives the first state close as much desired to the second one with respect to the $H^3_{(0)}-$norm.
After, we estimated a neighborhood in $H^3_{(0)}$ of any bounded state where the local exact controllability is satisfied in a given time.
In conclusion, for any couple of bounded state, we have defined a dynamics steering the first one into the second in explicit time.

\medskip
Given two bounded states, every aspect of the dynamics driving the first one to the second is explicit (up to the control function ruling the very last part of the dynamics).
Nevertheless, the estimates introduced in the work are far from being optimal and one might be interested in optimizing them in order to study meaningful physical systems.
The first try is to repeat the steps of the proof of Theorem $\ref{mainteo}$ by considering, from the beginning, specific $B$, $j$ and $k$. However, other possible improvements are the following. 

\begin{itemize}
	\item For instance, Theorem $\ref{mainteo}$ is stated for the control function $u_n(t)=n^{-1}\cos((\lambda_k-\lambda_j)t)$ with $n,j,k\in\N^*$. However, this choice is arbitrary. Indeed, the provided theory is based on $\cite{chambrion2}$ that considers generic $\frac{2\pi}{|\lambda_k-\lambda_j|}-$periodic controls. By retracing the proof of Theorem $\ref{mainteo}$ with a different suitable control, sharper results may be obtained.
	
	\item In Proposition $\ref{lclneigh}$, the controllability may be obtained in a larger neighborhood. Instead of Ingham's Theorem, one may use the \virgolette{Haraux's Theorem} $\cite[Theorem\  4.6]{Ing}$ and change the time $T=\frac{4}{\pi}$.
	
\end{itemize}

\medskip
\noindent
{\bf Acknowledgments.} The author thanks Thomas Chambrion for suggesting him the problem and for the explanations provided on the work $\cite{chambrion2}$.  He is also grateful to the colleagues Nabile Boussa\"id, Lorenzo Tentarelli and Riccardo Adami for the fruitful discussions.
This work has been partially supported by the ISDEEC project by ANR-16-CE40-0013.

\appendix

\section{Appendix: Explicit controls and times for the global approximate controllability}\label{pallose}
For $j,k\in\N^*$, let $T^*$, $u_n$, $b$, $E(j,k)$ and $C '$ be defined in $(\ref{definizioni})$. We denote
$$T=\frac{2\pi}{|\lambda_k-\lambda_j|},\ \ \ \ \ \ \ \ I=\frac{4}{|\lambda_k-\lambda_j|},\ \ \ \ \ \ \ \ K=\frac{2}{|B_{j,k}|},$$
In the following proposition, we provide a global approximate controllability result with explicit controls and times with respect to the $\Hi$-norm.

\begin{prop}\label{approssimaH}
	Let $B$ satisfy Assumptions I. For every $j,k\in\N^*$, $j\neq k,$ and $n\in\N^*$ such that
	\begin{equation}\label{cond1}
	n\geq\frac{3(1+C')|B_{j,k}|^{-1}\iii B\iii^2}{|k^2-j^2|},\end{equation}
	there exist $T_{n}\in (nT^*-T,nT^*+T)$ and $\theta\in\R$ such that
	$$\|\G_{T_n}^{u_n}\phi_j-e^{i\theta}\phi_k\|^2_{\Hi}\leq{\frac{3^2
			|B_{j,k}|^{-1}(1+C')\iii B\iii^2}{n|k^2-j^2|}}.$$
\end{prop} 

\begin{proof}
	Thanks to $\cite[Proposition \ 6]{chambrion2}$, for any $n\in\N^*$, there exists $T_{n}\in (nT^*-T,nT^*+T)$ such that
	\begin{equation}\label{3linee}\frac{1-|\la\phi_k,\G_{T_{n}}^{u_n}\phi_j\ra|}{1+2K\iii B\iii}\leq\frac{(1+C')\iii(\phi_j\la\phi_j,\cdot\ra+\phi_k\la\phi_k,\cdot\ra) B\iii I}{n}.\end{equation}
	We underline that the definition of $T^*$ given in $\cite[Proposition \ 6]{chambrion2}$ is incorrect and our formulation follows from $\cite[Proposition \ 2]{chambrion2}$. As a consequence of $(\ref{3linee})$, $1-|\la\phi_k,\G_{T_n}^{u_n}\phi_j\ra|\leq\frac{(1+2K\iii B\iii)(1+C')\iii B\iii I}{n}=:R_n$ and 
	\begin{equation*}\begin{split}
	&\sum_{l\neq k}|\la\phi_l,\G_{T_n}^{u_n}\phi_j\ra-\la\phi_l,\phi_k\ra|^2
	\leq\big(1-|\la\phi_k,\G_{T_n}^{u_n}\phi_j\ra|\big)\big(1+|\la\phi_k,\G_{T_n}^{u_n}\phi_j\ra|\big)\leq 2 R_n.\end{split}\end{equation*}
	Afterwards, fixed $n\in\N^*$, there exists $\theta\in\R$ (depending on $n$) such that $|\la\phi_k,e^{i\theta}\phi_k\ra-\la\phi_k,\G_{T_n}^{u_n}\phi_j\ra|^2\leq R_n^2$ and $R'_n:=\|e^{i\theta}\phi_k-\G_{T_n}^{u_n} \phi_j\|^2\leq{2R_n+R_n^2}.$
	As $|B_{j,k}|^{-1} \iii B\iii=\frac{\iii B\iii}{|\la\phi_j,B\phi_k\ra|}\geq 1$, we have
	\begin{equation*}\begin{split}
	&R_n=\frac{(1+2K\iii B\iii)(1+C')\iii B\iii I}{n}\leq\frac{3(1+C')|B_{j,k}|^{-1}\iii B\iii^2}{n|k^2-j^2|}.
	\end{split}\end{equation*}
	If $n\geq\frac{3(1+C')|B_{j,k}|^{-1}\iii B\iii^2}{|k^2-j^2|}$ for $j\neq k$, then $R_n\leq 1$, $R_n^2\leq R_n$ and
	\begin{equation}\label{cristoll}\begin{split}
	\|e^{i\theta}\phi_k-\G_{T_n}^{u_n}\phi_j\|^2\leq{2{R_n}+R_n^2}\leq{3R_n}\leq{\frac{3^2
			|B_{j,k}|^{-1}(1+C')\iii B\iii^2}{n|k^2-j^2|}}.\qedhere
	\end{split}
	\end{equation}
\end{proof}

\begin{prop}\label{approssimaH3}
	Let $B$ satisfy Assumptions I. Let $j,k\in\N^*$, $j\neq k,$ and $n\in\N^*$ verify the hypothesis of Theorem $\ref{mainteo}$. There exists $T_{n}\in (nT^*-T,nT^*+T)$ and $\theta\in\R$ such that
	\begin{equation*}\begin{split}
	\|\G_{T_n}^{u_n}\phi_j-e^{i\theta}\phi_k\|_{(3)}^8 \leq& \frac{2^{12} 3^{2}\pi^{24}|k^2-j^2|^5\max\{j,k\}^{24}(1+C')}{|B_{j,k}|^{7}n}\cdot\\
	&\frac{e^{\frac{6\iii B\iii_{(2)}}{|B_{j,k}|}}\iii B\iii_{(2)}^6\iii B\iii^2}{|B_{j,k}|^{7}n}.\\
	\end{split}\end{equation*}
\end{prop} 

\begin{proof}
	
	\needspace{3\baselineskip}
	
	\noindent
	{\bf 1) Propagation of regularity from $H^2_{(0)}$ to $H^4_{(0)}$:} We show that the propagator $\G_T^u$ preserves $H^4_{(0)}$ and $B\in L(H^2_{(0)})$. Let us denote
	$$\|f\|_{BV(T)}:=\|f\|_{BV((0,T),\R)}=\sup_{\{t_j\}_{ j\leq n}\in P}\sum_{j=1}^n |f(t_j)-f(t_{j-1})|$$
	for $f\in BV((0,T),\R)$, where $P$ is the set of the partitions of $(0,T)$ such that $t_0=0<t_1<...<t_n=T.$	
	Fixed $n\in\N^*$, we denote $$\lambda>0,\ \ \ \ \tilde\lambda=\lambda+\frac{\iii B\iii_{(2)}}{n},\ \ \ \ \ \widehat H^4_{(0)}:=D(A(i\tilde\lambda-A)).$$ We refer to $\cite{kato1}$ and we prove that the propagator $U^{u_n}_t$ generated by 
	$$A+u_n(t)B-i\|u_n\|_{L^\infty((0,T),\R)}\iii B\iii_{(2)}$$
	satisfies the condition $\|U^{u_n}_t\psi\|_{(4)}\leq C\|\psi\|_{(4)}$ for every $\psi\in H^4_{(0)}$ and suitable $C>0$. Indeed, $-i(A+u_n(t)B-i\|u_n\|_{L^\infty((0,T),\R)}\iii B\iii_{(2)})$ is not just dissipative in $H^2_{(0)}$, but also maximal dissipative thanks to Kato-Rellich's Theorem $\cite[Theorem\ 1.4.2]{relli}$. Now, Hille-Yosida Theorem implies that the semi-group generated by $-i(A+u_n(t)B-i\|u_n\|_{L^\infty((0,T),\R)}\iii B\iii_{(2)})$ is a semi-group of contraction and the techniques adopted in the proofs of $\cite[Theorem\ 2;\ Theorem \ 3]{kato1}$ are valid.
	As $\widetilde\lambda\geq \iii B\iii_{(2)}/n$, we have $$\iii u_n(t)B(i\tilde \lambda-A)^{-1}\iii_{(2)}\leq {\iii B\iii_{(2)}\iii( i\tilde\lambda- A)^{-1}\iii_{(2)}}{n^{-1}}\leq\frac{\iii B\iii_{(2)}}{n\widetilde \lambda}< 1.$$ We introduce $M:=\sup_{t\in [0,T_n]}\iii(i\tilde \lambda-A-u_n(t)B)^{-1}\iii_{L(H^2_{(0)},\widehat H^4_{(0)})}$ and
	\begin{equation}\label{rere}\begin{split}
	&M=\sup_{t\in [0,T_n]}\iii(I-u_n(t)B(i\tilde\lambda-A)^{-1})^{-1}\iii_{(2)}\\
	&=\sup_{t\in [0,T_n]}\iii\sum_{l=1}^{+\infty}(u_n(t)B(i\tilde \lambda-A)^{-1})^l\iii_{(2)}
	=\frac{n\widetilde\lambda}{n\widetilde\lambda-\iii B\iii_{(2)}}.\\
	\end{split}\end{equation} As $\|k+f(\cdot)\|_{BV((0,T),\R)}=\|f\|_{BV((0,T),\R)}$ for $f\in BV((0,T),\R)$ and $k\in\R$,
	\begin{equation*}\begin{split}N&:=\iii i\tilde \lambda-A-u_n(\cdot)B\iii_{BV\big([0,T_n],L(\widehat H^4_{(0)},H^2_{(0)})\big)}=n^{-1}\|u\|_{BV(T_n)} \iii B\iii_{L(\widehat H^4_{(0)},H^2_{(0)})}.\end{split}\end{equation*}
	Now, as $\|(A-i\tilde\lambda)\psi \|_{(2)}^2=\|A\psi \|_{(2)}^2+\tilde\lambda^2 \|\psi\|_{(2)}^2$, for every $\psi \in \widehat H^4_{(0)}$,
	$$	\|B\psi\|_{(2)}^2\leq \widetilde\lambda^{-2}\iii B\iii_{(2)}^2\big(\|A\psi \|_{(2)}^2+\tilde\lambda^2 \|\psi\|_{(2)}^2\big)
	\leq \widetilde\lambda^{-2}\iii B\iii_{(2)}^2\|\psi \|_{\widehat H^4_{(0)}}^2$$
	and $N\leq \frac{\iii B\iii_{(2)}\|u\|_{BV(T_n)}}{\widetilde\lambda n}$.
	Thanks to $\cite[Section\ 3.10]{kato1}$, there holds
	\begin{equation*}\begin{split}&\|(A+u_n(T_n)B-i\tilde\lambda )U_{T_n}^{u_n} \phi_j\|_{(2)}\leq Me^{MN}\|(A-i\tilde\lambda) \phi_j\|_{(2)}\leq Me^{MN}(\pi^2+{\tilde\lambda})j^4.\end{split}\end{equation*}
	In addition, thanks to the relation $(\ref{rere})$,
	\begin{equation*}\begin{split}
	&\iii A(A+u_n(T_n)B-i\tilde\lambda)^{-1}\iii_{(2)}\leq M +\iii\tilde\lambda(A-i\tilde\lambda)^{-1} \iii_{(2)} M\leq 2M.\\
	\end{split}\end{equation*}
	For every $j\in\N^*$, we know that
	$$\|\G_{T_n}^{u_n}\phi_j\|_{(4)}\leq e^{\frac{T_n}{n}\iii B\iii_{(2)}}\|U_t^{u_n}\phi_j\|_{(4)}=  e^{\frac{T_n}{n}\iii B\iii_{(2)}}2M^2e^{MN}(\pi^2+{\tilde\lambda})j^4. $$
	Now, $MN\leq \frac{\iii B\iii_{(2)}\|u\|_{BV(T_n)}}{n\widetilde\lambda-\iii B\iii_{(2)}}= \frac{\iii B\iii_{(2)}\|u\|_{BV(T_n)}}{n\lambda}$ and, if we choose $\lambda=\frac{\iii B\iii_{(2)}\|u\|_{BV(T_n)}}{n}$, then $MN\leq 1$ and $\widetilde \lambda=\frac{\iii B\iii_{(2)}(\|u\|_{BV(T_n)}+1)}{n}$. Now, $u$ is periodic and its total variation in a time interval long half-period is $2$. We compute $d$ the quarters of period for ${u}$ contained in $[0,nT^*+T]$ and  $$u_n(nT^*+T)=\frac{1}{n}\cos\big(\pi^2(k^2-j^2)(nT^*+T)\big)\ \ \ \Rightarrow \ \ \ d=\big(\pi^2|k^2-j^2|(nT^*+T)\big)\frac{2}{\pi}.$$
	As $d=2(n\pi^2|k^2-j^2|+{2}|B_{k,j}|)|B_{k,j}|^{-1}$, for the chosen $n$, we have $$\|u\|_{BV(T_n)}\leq\|u\|_{BV(nT^*+T)}\leq d+1\leq 2(n\pi^2|k^2-j^2|+{4}|B_{k,j}|)|B_{k,j}|^{-1},$$
	$$M^2=\Bigg(\frac{\iii B\iii_{(2)}(\|u\|_{BV(T_n)}+1)}{\iii B\iii_{(2)}\|u\|_{BV(T_n)}}\Bigg)^2\leq\Bigg(\frac{2(n\pi^2|k^2-j^2|+{6}|B_{k,j}|)|B_{k,j}|^{-1}}{2(n\pi^2|k^2-j^2|-{4}|B_{k,j}|)|B_{k,j}|^{-1}}\Bigg)^2,$$
	$$\pi^2+{\tilde\lambda}= \pi^2+\frac{\iii B\iii_{(2)}(\|u\|_{BV(T_n)}+1)}{n},$$
	\begin{equation}\label{bombissima1}\begin{split}
	\|\G_{T_n}^{u_n}\phi_j\|_{(4)}&\leq e^{\frac{T_n}{n}\iii B\iii_{(2)}}2\Big(\frac{\iii B\iii_{(2)}(\|u\|_{BV(T_n)}+1)}{\iii B\iii_{(2)}\|u\|_{BV(T_n)}}\Big)^2e (\pi^2+{\tilde\lambda})j^4\\
	&\leq e^{\frac{\iii B\iii_{(2)}}{|B_{j,k}|}+\frac{2\iii B\iii_{(2)}}{n\pi|k^2-j^2|}+1}2 \Big(\frac{\iii B\iii_{(2)}(\|u\|_{BV(T_n)}+1)}{\iii B\iii_{(2)}\|u\|_{BV(T_n)}}\Big)^2 (\pi^2+{\tilde\lambda})j^4.\\\end{split}\end{equation}
	\medskip
	
	\needspace{3\baselineskip}
	
	\noindent
	{\bf 2) Conclusion:} Let $f_n:=e^{i\theta}\phi_k-\G_{T_n}^{u_n} \phi_j$. First, we point out that, for every $s>0$,  we have $\|f_n\|^2_{(s)}\leq( k^s+\|\G_{T_n}^{u_n} \phi_j\|_{(s)})^2.$ As $\phi_j,\phi_k\in H^s_{(0)}$, for every $s>0$, the point {\bf 1)} ensures that $\G_T^u\phi_j$ and $\G_T^u\phi_j$ belong to $H^4_{(0)}$ for $u\in BV(0,T)$. Thanks to the Cauchy-Schwarz inequality,
	$$\|A^\frac{3}{2}f_n\|^4=\big(\la A^\frac{3}{2}f_n,A^\frac{3}{2}f_n\ra\big)^2\leq\big(\la A^2f_n, Af_n\ra\big)^2\leq\|A^2f_n\|^2\|Af_n\|^2,$$
	$$\|Af_n\|^2=\la Af_n, Af_n\ra\leq\|A^2f_n\|\|f_n\|,\ \ \ \ \ \Longrightarrow \ \ \ \ \|f_n\|_{(3)}^8\leq\|f_n\|^2\|f_n\|_{(4)}^6.$$
	For $R_n$ defined in the proof of Proposition $\ref{approssimaH}$, the relation $(\ref{bombissima1})$ implies
	\begin{equation*}\begin{split}
	&\|f_n\|_{(3)}^8 \leq  3R_n(\|\G_{T_n}^{u_n}\phi_j\|_{(4)}+k^4)^6\leq\frac{3^2
		|B_{j,k}|^{-1}(1+C')\iii B\iii^2}{n|k^2-j^2|}\cdot\\
	&\cdot\Big(e^{\frac{\iii B\iii_{(2)}}{|B_{j,k}|}+\frac{2\iii B\iii_{(2)}}{n\pi|k^2-j^2|}+1}2 \Big(\frac{\iii B\iii_{(2)}(\|u\|_{BV(T_n)}+1)}{\iii B\iii_{(2)}\|u\|_{BV(T_n)}}\Big)^2 (\pi^2+{\tilde\lambda})j^4+k^4\Big)^6\\
	&\leq \frac{2^{12} 3^{2}\pi^{24}(1+C')e^{\frac{6\iii B\iii_{(2)}}{|B_{j,k}|}} \iii B\iii_{(2)}^6\iii B\iii^2|k^2-j^2|^5\max\{j,k\}^{24}}{n|B_{j,k}|^{7}}.\qedhere
	\end{split}\end{equation*}
\end{proof}

\begin{prop}\label{approssimaH3T}
	Let $B$ satisfy Assumptions I and $n\in\N^*$ introduced in Theorem\ $\ref{mainteo}$. For every $j,k\in\N^*$ such that $j\neq k,$ there exists $\theta\in\R$ so that 
	\begin{equation*}\begin{split}
	\|\G_{nT^*}^{u_n}\phi_j-e^{i\theta}\phi_k\|_{(3)}^8\leq&\frac{2^{20} 3^{2}\pi^{24}e^{\frac{6\iii B\iii_{(2)}}{B_{j,k}}}|k^2-j^2|^5
		\max\{j,k\}^{24}}{|B_{j,k}|^{7}n}\cdot\\
	&\frac{(1+C') \iii B\iii_{(2)}^6\iii B\iii \max\{\iii B\iii,\iii B\iii_3\}}{|B_{j,k}|^{7}n}.\\
	\end{split}\end{equation*}
\end{prop} 
\begin{proof}
	We notice that the hypotheses of Proposition $\ref{approssimaH3}$ are verified. We estimate $\sup_{t\in [nT^*-T,nT^*+T]}\|\G_{t}^{u_n}\phi_j-\G_{T_n}^{u_n}\phi_j\|_{(3)}$ and we consider the arguments leading to $(\ref{bombissima1})$. The uniformly bounded constant $C(\cdot)$ is increasing and $(\ref{orcod})$ implies $\sup_{t\in [nT^*-T,nT^*+T]}C(|t-T_n|)\leq C(2T)\leq C({4}/{\pi})=\frac{24\sqrt{2}}{\pi^{4}}.$
	Thanks to Proposition $\ref{laura}$ and Remark $\ref{lauraa}$,
	\begin{equation*}\begin{split}
	&\sup\Big\{\sup_{t\in [nT^*-T,T_n]}\|\G_{t}^{u_n}\phi_j-\G_{T_n-t}^{u_n}\G_{t}^{u_n}\phi_j\|_{(3)},\sup_{t\in [T_n,nT^*+T]}\|\G_{t-T_n}^{u_n}\G_{T_n}^{u_n}\phi_j-\G_{T_n}^{u_n}\phi_j\|_{(3)}\Big\}\\
	&\leq C\Big(\frac{4}{\pi}\Big)\iii B\iii_3 \int^{nT^*+T}_{nT^*-T}|u_n(s)|ds\ \sup\Big\{\|\G_{T_n}^{u_n}\phi_j\|_{(4)},\sup_{t\in [nT^*-T,T_n]}\|\G_{t}^{u_n}\phi_j\|_{(4)}\Big\}.\\
	\end{split}
	\end{equation*}
	The techniques adopted in the proof of Proposition $\ref{approssimaH3}$ lead to 
	\begin{equation*}\begin{split}
	&\sup_{t\in [nT^*-T,T_n]}\|\G_{t}^{u_n}\phi_j\|_{(4)}
	\leq 2^2e^{\frac{\iii B\iii_{(2)}}{|B_{j,k}|}} \iii B\iii_{(2)}\pi^3|k^2-j^2||B_{k,j}|^{-1}j^4,\\
	\end{split}\end{equation*}
	which implies
	\begin{equation*}\begin{split}
	&\sup_{t\in (nT^*-T,nT^*+T)}\|\G_{t}^{u_n}\phi_j-\G_{T_n}^{u_n}\phi_j\|_{(3)}\leq C\Big(\frac{4}{\pi}\Big)\iii B\iii_3 \frac{2T}{n} \pi^{5}2^2e^{\frac{\iii B\iii_{(2)}}{|B_{j,k}|}} \iii B\iii_{(2)}\\
	&|k^2-j^2||B_{k,j}|^{-1}j^4\leq 2^5\ 3\ \sqrt{2}e^{\frac{\iii B\iii_{(2)}}{|B_{j,k}|}}\frac{\iii B\iii_3 }{n} 
	\iii B\iii_{(2)}|B_{j,k}|^{-1}j^4.\\
	\end{split}\end{equation*}
	Now, for $R''_n:=\|\G_{nT^*}^{u_n}\phi_j-e^{i\theta}\phi_k\|_{(3)}^8$, 
	\begin{equation}\label{R}\begin{split}
	&R''_n \leq 2^7\sup_{t\in (nT^*-T,nT^*+T)}\|\G_{t}^{u_n}\phi_j-\G_{T_n}^{u_n}\phi_j\|_{(3)}^8 +2^7\|\G_{T_n}^{u_n}\phi_j-e^{i\theta}\phi_k\|_{(3)}^8\\
	&\leq 2^7\Big(2^5\ 3\ \sqrt{2}e^{\frac{\iii B\iii_{(2)}}{|B_{j,k}|}}\frac{\iii B\iii_3 }{n} 
	\iii B\iii_{(2)}|B_{j,k}|^{-1}j^4\Big)^8+2^7\|f_n\|_{(3)}^8.\\
	\end{split}\end{equation}
	Now, $\iii B\iii,\iii B\iii_{(2)}\geq|B_{j,k}|$ for $j,k\in\N^*$. For the chosen $n\in\N^*$, we have 
	\begin{equation*}\begin{split}
	R''_n &\leq 2^7
	\left(2^5\ 3\ \sqrt{2}e^{\frac{\iii B\iii_{(2)}}{|B_{j,k}|}}\frac{\iii B\iii_3 }{n} 
	\iii B\iii_{(2)}|B_{j,k}|^{-1}j^4+\|f_n\|_{(3)}^8\right)\\
	&\leq\frac{2^{20} 3^{2}\pi^{24}(1+C')e^{\frac{6\iii B\iii_{(2)}}{|B_{j,k}|}} \iii B\iii_{(2)}^6\iii B\iii^2|k^2-j^2|^5\max\{j,k\}^{24}}{n|B_{j,k}|^{7}}.\qedhere
	\end{split}\end{equation*}
\end{proof}

\section{Appendix: Moment problem}\label{appmome1}

In this appendix, we briefly adapt some results concerning the solvability of the moment problem (as the relation $(\ref{mome})$). Let $\cite[Proposition\ 19; \ 2)]{laurent}$ be satisfied and $\{f_k\}_{k\in\Z}$ be a Riesz Basis (see $\cite[Definition\ 2]{laurent}$) in 
$$X=\overline{span\{f_k:\ k\in\Z\}}^{\ \Hi}\subseteq \Hi,$$
with $\Hi$ an Hilbert space. For $\{v_k\}_{k\in\Z}$ the unique biorthogonal family to $\{f_k\}_{k\in\Z}$ ($\cite[Remark\ 7]{laurent}$), $\{{v}_k\}_{k\in\Z}$ is also a Riesz Basis of $X$ ($\cite[Remark\ 9]{laurent}$). If $\{f_k\}_{k\in\Z}$ is the image of an orthonormal family $\{e_k\}_{k\in\Z}\subset\Hi$ by an isomorphism $V:\Hi\rightarrow \Hi$, then $\{v_k\}_{k\in\Z}$ is the image of $\{e_k\}_{k\in\Z}\subset\Hi$ by the isomorphism $(V^*)^{-1}:\Hi\rightarrow \Hi$. Indeed, 
$$\delta_{k,j}=\la v_k,f_j\ra_{\Hi}=\la v_k,V(e_j)\ra_{\Hi}=\la V^*(v_k),e_j\ra_{\Hi},\ \ \ \ \ \forall k,j\in\Z,$$
that implies $(V^*)^{-1}(e_k)=v_k$ for every $k\in\Z$. We point out that in $\cite[relation\ (71)]{laurent}$ there is a misprint as there exist $C_1,C_2>0$ such that
\begin{equation}\label{cris}C_1\sum_{k\in \Z}|x_k|^2\leq\|u\|_{\Hi}^2\leq C_2\sum_{k\in \Z}|x_k|^2,\end{equation}
for every $u(t)=\sum_{k\in  \Z}x_k f_k(t)$ with $\{x_k\}_{k\in\N^*}\in\ell^2(\C)$.
The arguments of the proof of $\cite[Proposition\ 19; \ 2)]{laurent}$ and the relations
$$(V^*)^{-1}=(V^{-1})^*,\ \ \ \iii V^*\iii_{L(\Hi)}=\iii V\iii_{L(\Hi)},\ \  \ \iii (V^{-1})^*\iii_{L(\Hi)}=\iii V^{-1}\iii_{L(\Hi)}$$
implies that, for every $u(t)=\sum_{k\in  \Z}x_k v_k(t)$ with $\{x_k\}_{k\in\N^*}\ell^2(\C)$, we have $$C_2^{-2}\sum_{k\in \Z}|x_k|^2\leq\|u\|_{\Hi}^2\leq C_1^{-2}\sum_{k\in \Z}|x_k|^2.$$
The constants $C_1, C_2>0$ are the same appearing in $(\ref{cris})$.
Moreover, for every $u\in X$, we know that $u=\sum_{k\in\Z}v_k \la f_k,u\ra_{\Hi}$ since $\{f_k\}_{k\in\Z}$ and $\{v_k\}_{k\in\Z}$ are reciprocally biorthonormal (see $\cite[Remark\ 9]{laurent}$) and
\begin{equation}\label{weirdo}C_2^{-1}\Big(\sum_{k\in \Z}|\la f_k,u\ra_{\Hi}|^2\Big)^{\frac{1}{2}}\leq\|u\|_{\Hi}\leq C_1^{-1}\Big(\sum_{k\in \Z}|\la f_k,u\ra_{\Hi}|^2\Big)^{\frac{1}{2}}.\end{equation}

\begin{osss}\label{pu}
	
	Let $\{\lambda_k\}_{k\in\N^*}\subset\R^+$ be so that $\Gi:=\inf_{k\neq j}|\omega_{k}-\omega_{j}|> 0$. Thanks to the Ingham's Theorem $\cite[Theorem \ 4.3]{Ing}$, for $T>\frac{2\pi}{\Gi}$, the family of functions $\{e^{i\lambda_k (\cdot)}\}_{k\in\Z}$ is a Riesz Basis in $X=\overline{span\{e^{i\lambda_k (\cdot)}:\ k\in\Z\}}^{\ L^2}$. In the current remark, we consider $\Hi=L^2((0,T),\C).$ From \eqref{weirdo}, we have
	\begin{equation}\label{weirdost}C_2^{-1}\Big(\sum_{k\in \Z}|\la e^{i\lambda_k (\cdot)},u\ra_{\Hi}|^2\Big)^{\frac{1}{2}}\leq\|u\|_{\Hi}\leq C_1^{-1}\Big(\sum_{k\in \Z}|\la e^{i\lambda_k (\cdot)},u\ra_{\Hi}|^2\Big)^{\frac{1}{2}}.\end{equation}
	The relation $(\ref{weirdost})$ ensures that $F:u\in X\longmapsto \big\{\la e^{i\lambda_k (\cdot)},u\ra_{\Hi}\big\}_{k\in\Z}\in\ell^2(\C)$ is injective. Let $\{v_k\}_{k\in\Z}$ be the unique biorthogonal family to $\{e^{i\lambda_k (\cdot)}\}_{k\in\Z}$. The surjectivity of the map $F$ follows as, for every $\{x_k\}_{k\in\Z}\in \ell^2(\C)$ and $N\in\N^*$, $$u^N(t)= \sum_{k\leq N}v_k x_k\in X,\ \ \ \ \ \big\{\la e^{i\lambda_k (\cdot)},u^N\ra_{\Hi}\big\}_{k\leq N}=\{x_k\}_{k\leq N}.$$ Since $\big\{\{x_k\}_{k\leq N}\big\}_{N\in\N^*}$ is a Cauchy sequence, $\{u^N\}_{N\in\N^*}$ is also a Cauchy sequence in $L^2((0,T),\R)$ thanks to $(\ref{weirdost})$. The completeness of $X$ implies
	$$u(t)= \sum_{k\in\N^*} v_k x_k\in X,\ \ \ \ \ \big\{\la e^{i\lambda_k (\cdot)},u\ra_{\Hi}\big\}_{k\in \N^*}=\{x_k\}_{k\in\N^*}.$$ Thus, $F:u\in X\longmapsto \big\{\la e^{i\lambda_k (\cdot)},u\ra_{\Hi}\big\}_{k\in\Z}\in\ell^2(\C)$ is an homeomorphism and, for every $\{x_k\}_{k\in\Z}\in \ell^2(\C)$, there exists a unique $u\in X$ such that
	\begin{equation*}x_k=\int_0^Tu(s)e^{-i\lambda_k s}ds,\ \ \ \ \ \ \ \forall k\in\Z.\end{equation*}
\end{osss}

\begin{lemma}\label{solve112}
	Let $\{\mu_k\}_{k\in\N^*}=\{\pi^2(k^2-l^2)\}_{k\in\N^*}$ for $l\in\N^*$ such that
	\begin{equation}\label{sborrato}\mu_{-k}=\pi^2(k^2-l^2)\neq \pi^2(l^2-j^2)=-\mu_j,\ \ \ \ \ \ \ \forall k,j\in\N^*.\end{equation}
	For $T>2/\pi$, for every $\{x_k\}_{k\in\N^*}\in\ell^2(\C)$ such that $x_{l}\in\R$, 
	\begin{equation*}\begin{split}
	\exists u\in L^2((0,T),\R)\ \ \ \ \  :\ \  \ \ \ x_k=\int_0^Tu(s)e^{i\mu_k s}ds, \ \ \  \ \ \ \forall k\in\N^*.\\
	\end{split}\end{equation*}
	In addition, there exists $X\subseteq L^2((0,T),\R)$ such that the map $$\widetilde J:u\in X\longmapsto \{\la u,e^{i\mu_k(\cdot)}\ra\}_{k\in\N^*}\in \{\{x_k\}_{k\in\N^*}\in\ell^2(\C):\ x_{l}\in\R\}$$ is an homeomorphism.
\end{lemma}
\begin{proof}
	For $k> 0$, we call $\omega_k=-\mu_{k}$, while we impose $\omega_k=\mu_{-k}$ for $k<0$ and $k\neq -l$. 
	We denote $\Z^*=\Z\setminus\{0\}$. 
	The sequence $\{\omega_k\}_{k\in \Z^*\setminus\{-l\}}$ satisfies the hypotheses of $\cite[Theorem \ 4.3]{Ing}$ as $\Gi:=\inf_{k\neq j}|\omega_{k}-\omega_{j}|\geq\pi^2$ thanks to the relation $(\ref{sborrato})$. Thus, Remark $\ref{pu}$ is valid.
	Given $\{x_k\}_{k\in\N^*}\in \ell^2(\C)$, we introduce $\{\widetilde x_k\}_{k\in \Z^*\setminus\{-l\}}\in \ell^2(\C)$ such that $\widetilde x_k=x_{k}$ for $k> 0$, while $\widetilde x_k=\overline x_{-k}$ for $k<0$ and $k\neq -l$. 
	For $T>2\pi/\Gi$, there exists $u\in L^2((0,T),\C)$ so that $ \widetilde x_k=\int_0^Tu(s)e^{-i\omega_k s}ds$ for each $k\in \Z^*\setminus\{-l\}$. Then
	\begin{equation*}\begin{split}\begin{cases}
	x_k=\int_0^Tu(s)e^{i\mu_k s}ds=\int_0^T\overline{u}(s)e^{i\mu_k s}ds,\ \ \ \ \ \ \ \ & k\in\N^*\setminus\{l\},\\
	x_k=\int_0^Tu(s)ds,\ \ \ \ \ \ \ \ & k= l,\\
	\end{cases}
	\end{split}\end{equation*}
	which implies that $u$ is real when $x_{l}\in\R$. For $\{v_k\}_{k\in\N^*}$ the biorthogonal family to $\{e^{i\mu_k(\cdot)}\}_{k\in\N^*}$,
	we have $v_{l}\in\R$ and $\{\overline{v}_k\}_{k\in\N^*}$ is the biorthogonal family to $\{e^{-i\mu_k(\cdot)}\}_{k\in\N^*}$. Thus, $u(t)=\sum_{k\in\N^*} \widetilde{x}_k v_k(t)+\sum_{k\in\N^*\setminus\{l\}} \widetilde{x}_{-k} \overline{v_k}(t)= x_{l}v_{l}(t)+2\sum_{k\in\N^*\setminus{l}} \Re(x_{k} {v_k}(t))$ and $(\ref{weirdost})$ leads to
	\begin{equation}\label{weirdos}\begin{split}
	&C_2^{-1}\Big(\sum_{k\in \N^*}|x_k|^2\Big)^\frac{1}{2}\leq\|u\|_{2}\leq 2C_1^{-1}\Big(\sum_{k\in \N^*}|x_k|^2\Big)^\frac{1}{2}.\end{split}\end{equation}
	For ${\bf x}:=\{x_k\}_{k\in\Z^*\setminus\{-l\}}$ belonging to $\ell^2_{l}(\C):=\{\{x_k\}_{k\in\Z^*\setminus\{-l\}}: \{x_k\}_{k\in\N^*}\in\ell^2(\C);\ x_{-k}=\overline{x_k},\ \forall k\in -\N^* \setminus\{-l\};\ x_{l}\in\R\}$, we define
	\begin{equation*}\begin{split}
	u_{\bf x}(t)=x_{l}v_{l}+2\sum_{k\in\N^*\setminus\{l\}} \Re(x_{k} {v_k}),\ \ \ \ \ \ \ \ 	X:=\{u_{\bf x}:\ {\bf x}\in \ell^2_{l}(\C) \}.\\
	\end{split}\end{equation*} 
	From $(\ref{weirdos})$, $J:u\in X\longmapsto \{\la u,e^{i\omega_k(\cdot)}\ra\}_{k\in\Z^*\setminus\{-l\}}\in \ell^2_{l}(\C)$ is an homeomorphism (for $\{\omega_k\}_{k\in\N^*}$ defined above), which implies the result.\qedhere
\end{proof}

\end{document}